\documentclass{llncs}
\usepackage{subfig}
\usepackage{amsfonts,amsmath,amstext,amssymb}

\usepackage{algorithm,algorithmicx}	
\usepackage{algpseudocode}			

\usepackage{enumerate}
\usepackage{paralist}

\usepackage{todonotes}
\usepackage{epstopdf}
\usepackage{color}

\usepackage{pgf}
\usepackage{tikz}
\usetikzlibrary{arrows,automata}
\usepackage[all]{xy}

\usepackage{hyperref}

\begin{document}

\title{D-Iteration: diffusion approach for solving PageRank }
\titlerunning{D-Iteration: a diffusion approach for solving PageRank }  
%
\author{Dohy Hong\inst{1} \and The Dang Huynh\inst{2}
\and Fabien Mathieu\inst{2}}
\authorrunning{Dohy Hong et al.} 
%
%
\institute{Samsung
\and
Alcatel-Lucent Bell Labs France}

\maketitle              

\begin{abstract}
In this paper we present a new method that can accelerate the computation of the PageRank importance vector. Our method, called D-Iteration (DI), is based on the decomposition of the matrix-vector product that can be seen as a fluid diffusion model and is potentially adapted to asynchronous implementation.
We give theoretical results about the convergence of our algorithm and we show through experimentations on a real Web graph that DI can improve the computation efficiency compared to other classical algorithm like Power Iteration, Gauss-Seidel or OPIC . 
\end{abstract}

\section{Introduction}

%
%
%
%
%
%
%
%

PageRank is a link analysis algorithm that has been initially introduced in \cite{BP99} and used by the Google Internet search engine. It assigns a numerical value to each element of a hyper-linked set of nodes, such as the World Wide Web. The algorithm may be applied to any collection of entities (nodes) that are linked through directional relationships. The numerical value assigned to each node called PageRank, is supposed to reflect the structural importance of nodes.

Although PageRank may today only be a small part of Google's ranking algorithm (the complete algorithm is obviously kept secret, but it seems to take into account hundreds of parameters, most of them been related to the user's profile), it stays appealing, especially in the research community, as it balances simplicity of definition, ranking efficiency and computational challenges.

Among these challenges are the growing size of the dataset (Web graphs can have tens of billions of nodes) and the dynamics ot the structure that requires frequent updates.

The complexity of computing the PageRank of a graph rapidly increases with the number of nodes, as it is equivalent to computing an eigenvector on some huge, matrix, and efficient and accurate methods to compute eigenvalues and eigenvectors of arbitrary matrices are in general a difficult problem. In the particular case of the PageRank equation, several specific solutions were proposed and analysed \cite{LM04,BM05} including the power method \cite{BP99} with adaptation \cite{ST03} or extrapolation \cite{TS03,CG03}, or the adaptive on-line method \cite{AP03}, etc.

In this paper we present some theoretical results of D-Iteration (DI): mathematical definition, convergence to the fixed point, measurement of the distance to the limit and solution for dynamic graph update issues. The results are validated on a real dataset.



The rest of the paper is organized as follows. Section~\ref{sec:the-pagerank-challenge} recalls the PageRank equation and exposes some of the main methods used to solve it. Section~\ref{sec:d-iteration-algorithm} formally defines the D-Iteration method and provides a few theoretical results on its convergence and use in a dynamic context.
Lastly, Section~\ref{sec:experiments} proposes a numerical evaluation of DI's performance and Section~\ref{sec:conclusion} concludes.

\section{The PageRank challenge}\label{sec:the-pagerank-challenge}


In this section, we shortly introduce the definition of PageRank along with some standard algorithms to compute it. For more detailed information, one may for example see \cite{LM04}.

\subsection{Definition}

The informal definition of PageRank is rather simple: it is an importance vector over Web pages such that important pages are referenced by important pages~\cite{BP98}.


More formally, let $ G=(V,E) $ be a weighted, oriented, graph. The size of $ G $ is $ n=|V|$ and $ w_{i,j}>0 $ is the weight of edge $ (i,j)\in E $.
$ G $ represents a set of nodes their (weighted, oriented) relationships. In \cite{BP98}, $ G $ was a Web graph, $ V $ representing web pages and $ E $ hyperlinks, but the principle applies to most structured sets.

Let $ P $ be a $ n\times n $ diffusion matrix defined by:

\begin{equation}
P_{i,j}=\left\{ 
\begin{array}{l}
\frac{w_{j,i}}{\sum_{(j,k)\in E}w_{j,k}} \text{ if $ (j,i) \in E$,}\\
0 \text{ otherwise.}
\end{array}
\right.
\label{eq:def-p}
\end{equation}


$ P $ is a left substochastic matrix, column $ i $ summing to $ 1 $ if node $ i $ has an outgoing edge, $ 0 $ if $ i $ is a dangling node. Note that:
\begin{itemize}
\item For unweighted graphs, the expression of $ P $ is simpler: for $ (j,i)\in E $, we just have $ P_{i,j}=1/\textrm{out}(j) $, where $\textrm{out}(j) $ is the out-degree of $ j $.
\item Some variants of PageRank require $ P $ to be stochastic.
For these variants, one usually pads the null columns of $ P $ with $ \frac{1}{n} $ (\emph{dangling nodes completion}).
\end{itemize}


$ P $ represents how importance flows from node to node. When it is stochastic, it represents the Markov chain over $ V $ implied by the edges $ E $. In that case, the PageRank can be defined as a stationary state of the Markov chain, that is a distribution $ x $ over $ V $ that verifies
\begin{equation}
x=Px\text{.}
\label{eq:pagerank-simple}
\end{equation}
Note that $ x $ is unique if $ G $ is strongly connected.



In practice, the following variant is used instead of \eqref{eq:pagerank-simple}:
\begin{equation}
x=dPx+(1-d)Z\text{,}
\label{eq:pagerank-pf}
\end{equation}
\noindent where $ 0<d<1 $ is called \emph{damping} (often set to $ 0.85 $), and $ Z $ a \emph{default distribution} over $ V $ (often set to the uniform distribution).

If $ P $ is stochastic, the solution $ x $ of \eqref{eq:pagerank-pf} is a distribution, which corresponds to the Markov chain defined by: with probability $ d $, use the Markov chain induced by $ P $; otherwise jump to a random node according to distribution $ Z $.

Introducing parameters $ d $ and $ Z $ has many advantages:
\begin{itemize}
\item It guarantees the existence and uniqueness of a solution for any substochastic matrix $ P $, without any assumption on $ G $.
\item It speeds up the PageRank computation (cf below).
\item Parameter $ d $ introduces some locality: influence of a node at distance $ k $ is reduced by a factor $ d^k $. This strengthens the impact of the local structure and mitigates the possibility of malicious PageRank alterations through techniques like links farm\cite{BM05}.
\item Parameter $ Z $ allows to customize the PageRank. For instance, one can concentrate the default importance on pages known to talk about some given topic to create a topic-sensitive PageRank\cite{HT02}.
\end{itemize}

\label{sec:pagerank}

In the rest of the paper, unless stated otherwise, we focus on solving \eqref{eq:pagerank-pf}. Writing the solution is straightforward:
\begin{equation}
x=(1-d)(I-dP)^{-1}Z\text{, where $ I $ is the identity matrix.}
\end{equation}

However, such a direct approach cannot be used due to the size of $ P $ that forbids an explicit computation of $ (I-dP)^{-1} $. Instead, one can use different iterative methods (cf \cite{LM04}). 

\subsection{Power Iteration}

The simplest approach is \textbf{Power Iteration (PI)}: starting from an initial guess vector $x_0$, the stationary PageRank vector is iteratively computed using \eqref{eq:pagerank-pf}:
\begin{equation}
x_{k+1}=dPx_{k}+(1-d)Z\text{,}
\end{equation}
until the change between two consecutive vectors is negligible. It is straightforward that the error decays by a factor at least $ d $ at each iteration (hence one of the interests of introducing a damping factor). PI requires to maintain two vectors $ x_k $ and $ x_{k+1} $.



\subsection{Gauss-Seidel}

The \textbf{Gauss-Seidel (GS)} applied to PageRank consists in using the updated entries of $ x_k $ as they are computed. During an iteration round, entries $ x_{k+1}(i) $ are computed from $ i=1 $ to $ i=n $ using:
\begin{equation}
x_{k+1}(i)=d\left(\sum_{j<i} P_{i,j}x_{k+1}(j)+\sum_{j\geq i} P_{i,j}x_k(j) \right)+(1-d)Z(i)\text{.}
\end{equation}

Thanks to the immediate update, one needs to maintain only one vector and the convergence is faster, typically by a factor 2 asymptotically. The main downside of the update mechanism is the necessity to access the entries in a round-robin fashion, which can cause problems in a distributed scenario.
Note that Gauss-Seidel belongs to a larger class of methods called Successive Overrelaxation (SOR), but other SOR variants are seldom used for Web PageRank computations\cite{SOR}.



\subsection{Online Page Importance Computation}
\label{sec:opic}

An algorithm very close to ours is the \textbf{Online Page Importance Computation (OPIC)} proposed in \cite{AP03}. Its core idea: most PageRank algorithms implicitly use a \emph{pull} approach, where the state of a node is updated according to the states of its incoming neighbors. By contrast, OPIC proposes a \emph{push} approach, where the state of a node is read and used to update the states of its outgoing neighbors. In details, OPIC focuses on solving \eqref{eq:pagerank-simple} for a modified graph $ G'=(V\cup z,E\cup J) $, where $ z $ is a virtual \emph{zap} node and $ J=\left(V\times z\right) \cup \left(z\times V\right) $ is all possible edges between $ V $ and $ z $. This was introduced to make $ P $ stochastic and irreducible, allowing \eqref{eq:pagerank-simple} to admit a unique solution.

OPIC algorithm works as follows: initially, each node receives some amount of \emph{fluid} (a non-negative number) and a null history record. A scheduler, which can be associated to a web crawler, iterates among the nodes. When a node $ i $ is selected, its fluid $ F(i) $ is, in order,
\begin{itemize}
\item credited to its history: $ H(i)=H(i)+F(i) $;
\item equally pushed to its neighbors: for all $ j $ that are outgoing neighbors of $ i $, $ F(j)=F(j)+\frac{F(i)}{\textrm{out}(i)} $;
\item cleared: $ F(i)=\frac{F(i)}{\textrm{out}(i)} $ if $ i $ has a loop, $ F(i)=0 $ otherwise.
\end{itemize}

It has been shown that as long as the scheduler is fair (i.e. any given node is selected infinitely often) then the history vector converges, up to normalization, to the desired solution \cite{AP03}.

The main advantage of OPIC is its flexibility. In particular, it is easy to adapt and incorporate to a continuous, possibly distributed, Web crawler, allowing to get a dynamic, lightweight, PageRank importance estimation. One drawback is that it is not designed to work with \eqref{eq:pagerank-pf}.

\subsection{Other Variants}

Due to lack of space, we only briefly describe a few other methods that have been proposed.

The Generalized Minimal Residual (GMRES) \cite{SS86} provides an approach using Arnoldi process to find the stationary vector in Krylov subspace. It is claimed to perform better than other algorithms in terms of iterations, but it requires more elementary computation steps per iteration.

Quadratic Extrapolation \cite{CG03} uses estimates of secondary eigenvectors of $ P $ to speed up the convergence of PI.

Lastly, a few methods have been proposed that take advantage of the locality of hyperlinks (most hyperlinks are internal to Web sites) to decompose PageRank computation into intra-site and extra-site values \cite{WD04,blockrank,flowrank}. These methods are especially adapted to distributed computation, but they usually only output an approximation of the solution of \eqref{eq:pagerank-pf}.
%
%
%
%
%
%
%
%

\section{D-Iteration Algorithm}\label{sec:d-iteration-algorithm}

Our proposal, D-Iteration, aims at solving Equation \eqref{eq:pagerank-pf} with an efficiency similar to Gauss-Seidel while keeping the scheduling flexibility offered by OPIC. This result in a fluid diffusion approach similar to OPIC with some damping added to the mix.

\subsection{Definition}

In the following, we assume that a deterministic or a random diffusion sequence $\mathcal{I} = \{i_1, i_2, ..., i_k,...\}$ with $i_k \in \{1,..,n\}$ is given. We only require that the number of occurrences of each value of $i_k$ in $\mathcal{I}$ is infinite (the scheduler $ \mathcal{I} $ is fair). $\mathcal{I}$ is not obliged to be fixed in advance but can be adjusted on the fly as long as the fairness property stands.


Like with OPIC, we have to deal with two variable vectors at the same time: a \emph{fluid vector} $F$, initially equal to $ (1-d)Z $ and a \emph{history vector} $H$, initially null for all nodes.
When a node is selected, its current fluid value is added to its history, then a \emph{fraction $ d $} of its fluid is equally pushed to its neighbors and its fluid value is cleared.

Formally, the fluid vector $ F $ associated to the scheduler $ \mathcal{I} $ is iteratively updated using:
\begin{eqnarray}
F_0 &=& (1-d) Z\text{,}\\
F_k &=& F_{k-1} + d F_{k-1}(i_k) P e_{i_k}  - F_{k-1}(i_k) e_{i_k}\text{,}\label{eq:defF}
\end{eqnarray}

where $e_{i_k}$ is the standard basis vector corresponding to $ i_k $.

The second term in \eqref{eq:defF} represents the damped diffusion and the third term clears the local fluid (up to loops). Similarly to OPIC, an iteration reads one value, $ F_{k-1}(i_k) $ and updates $ i_k $ and its outgoing neighbors. Note that F is always non-negative.

We also formally define the history vector $H$:

\begin{eqnarray}
H_0 &=& \vec{0}\text{,} \\
H_k &=& H_{k-1}+F_{k-1}(i_k)e_{i_k}\text{.}
\end{eqnarray}
By construction, $H_k$ is non-decreasing with $ k $.
%
%

\subsection{Convergence}

The following Theorem states the convergence of the D-Iteration algorithm:

\begin{theorem}
\label{thm:conv}
For any fair sequence $\mathcal{I}$, the history vector $H_k$ converges to the unique vector $x$ such that $x=dPx + (1-d)Z$:
$$\lim\limits_{k\rightarrow\infty}H_k=(1-d)(I-dP)^{-1}Z\text{.}$$

Moreover, we have
\begin{equation}\label{eq:SL}
|x - H_k| \leq  \frac{|F_k|}{1-d}\text{, where $|\cdot|$ is the $L_1$ norm.}
\end{equation}

\end{theorem}

\begin{proof}

We first prove the equality:
\begin{equation}\label{eq:HnFn}
H_k + F_k = F_0 + d P H_k.
\end{equation}

This is straightforward by induction:  \eqref{eq:HnFn} is true for $ k=0 $; assuming it is true for $k-1$, we have

$$
\begin{array}{lcl} 
H_{k} + F_{k} &=&
H_{k-1}+F_{k-1}(i_k)e_{i_k}+F_{k-1} + d F_{k-1}(i_k) P e_{i_k} - F_{k-1}(i_k)e_{i_k}\\
&=& F_0 + d P (H_{k-1}+  F_{k-1}(i_k)e_{i_k})=F_0 + d P H_k.
\end{array}
$$

From the equation \ref{eq:HnFn}, we have:

$$
\begin{array}{lcl}
H_k &=& (I-dP)^{-1}(F_0-F_k)\\
&=&  x - \sum_{i=0}^{\infty}d^iP^iF_k 
\end{array}
$$

Noticing that $ P $ is substochastic, we get 

$$
|x - H_k|=|\sum_{i=0}^{\infty}d^iP^iF_k|\leq \sum_{i=0}^{\infty}d^i|F_k|=\frac{|F_k|}{1-d}\text{.}
$$

All we need is to show that $|F_k|$ tends to 0. Notice that the total available fluid is non-increasing. That being said, the fluid of a given node is non-decreasing until it is scheduled, and when it is, a quantity $ (1-d) $ of it is ``lost'' due to the damping (or more if it is a dangling node). Given these two observations, let us consider a time $ k $ and another time $ k'>k $ such that all nodes have been scheduled at least once between $ k $ and $ k' $ (this is always feasible thanks to the fairness assumption). For each node, its fluid at the time of its first scheduling after $ k $ is greater than its fluid at time $ k $, so we have $|F_{k'}|\leq d |F_k|$. This is true for any $ k $ (including $ k' $) and concludes the proof.
%
\end{proof}

%

\subsection{Implicit completion}

\label{ss:distance}

Equation \eqref{eq:SL} is an equality if $P$ is stochastic (for example thanks to dangling node completion), in which case we have  $|P^i F_k| = |F_k|$. 

One may want to solve \eqref{eq:SL} for a completed $ P $. However, if the completion follows the distribution $ Z $, one observes that the result is just proportional to the one without completion, so this is just a matter of normalization.

It is more efficient  to perform the computation on a non-completed matrix (every time a dangling node is selected, all non-null entries of $ Z $ are updated if $ P $ is completed). The question is: can we control the convergence to the solution of the completed matrix while running the algorithm on the original one?

%
%

To address this problem precisely, we count the total amount of fluid that has left the system when a diffusion was applied on a dangling node. We call this quantity $l_k$ (at step $k$ of the DI). To compensate this loss and emulate completion, a quantity $d l_k Z$ should have been added to the initial fluid, leading to $(1-d+dl_k)Z$ instead of $(1-d)Z$. But then the fluid $dl_k Z$ would have produced after $k$ steps a leak $(d l_k^2/(1-d))Z$ on dangling nodes, which needs to be compensated\ldots 

In the end, the correction that is required to compensate the effect of dangling nodes on the residual fluid $|F_k|$ consists in replacing the initial condition $|F_0|=(1-d)$ by  $|F'_0|$ such that:
\begin{center}
$\begin{array}{lcl}|F'_0| &=& (1-d) + d l_k + d l_k\displaystyle\frac{d l_k}{1-d} + d l_k\left(\displaystyle\frac{d l_k}{1-d}\right)^2 + ...\\
&=& (1-d) + d l_k\displaystyle\sum_{n=0}^{\infty}\left(\frac{d l_k}{1-d}\right)^n \\
&=& \displaystyle\frac{(1-d)^2}{1-d - d l_k}. \end{array}$
\end{center}

As $|F'_0|/|F_0|=(1-d)/(1-d-d l_k)$, $H_k$ needs to be renormalized (multiplication) by $(1-d)/(1-d-d l_k)$ so that the exact $L_1$ distance between $ x $ and the normalized $H$ is equal to:
$$
|x-\displaystyle\frac{1-d}{1-d-d l_k}H_k| = \displaystyle\frac{|F_k|}{1-d-d l_k}.
$$

To summarize, we can run the algorithm on the original matrix using 
$ \frac{|F_k|}{1-d-d l_k} $ as a stopping condition that guarantees the precision of the normalized result.


\subsection{Schedulers}

The actual performance of DI is directly related to the scheduler used. A simple scheduler, which we call \texttt{DI-cyc}, is a Round-Robin (RR) one, where a given permutation of nodes is repeated as long as necessary.

\begin{theorem}
\label{thm:rr}
For any Round-Robin scheduler $\mathcal{I}$, we have: 
\begin{equation}\label{eq:rrconv}
|x - H_k| \leq  d^{\lfloor \frac k n \rfloor} \text{.}
\end{equation}

\end{theorem}

The proof is a direct application of the proof of Theorem \ref{thm:conv} considering successive sequences of $ n $ steps.

Theorem \ref{thm:rr} ensures that D-Iteration performs at least as well as the PI method: in both cases, after a round where all nodes have been processed once, the error id reduced by at least $ d $.

While the bound can be tight for specific situations (for instance a clockwise scheduler applied to a counterclockwise-oriented cycle), it is conservative in the sense that it ignores that some of the fluid can be pushed multiple times during a sequence of $ n $ steps. For that reason, and keeping in mind that D-Iteration is a \emph{push} version of the Gauss-Seidel method, we expect that  \texttt{DI-cyc} will perform more like Gauss-Seidel in practice.

However, one strength of D-Iteration is the freedom in the choice of a scheduler. As the convergence is controlled by the remaining fluid $|F_k|$, a good scheduler is one that makes the fluid vanish as quickly as possible. This disappearance is caused by two main parameters: the damping factor $d$ and dangling nodes (cf above). Reminding that at time $ k $ a quantity $(1-d)F_{k-1}(i_k)  $ vanishes through damping, a natural greedy strategy would consist in selecting at each step $ i_k =\textrm{argmax}_{i=1,\ldots,n}F_{k-1}(i)  $.

The main drawback of such a strategy is the expensive searching cost. Therefore we propose a simple heuristic called \texttt{DI-argmax} as follows: we use a RR scheduler, but at each iteration, we run through the scheduler until we find a node that possesses a fluid greater or equal to the average fluid in the whole graph. The advantage of this method is that nodes with relatively low fluids will be skipped, avoiding unprofitable updates operation, with a searching cost lower than picking up the best node at each iteration.

\begin{theorem}
\label{thm:argmax}
Using \texttt{DI-argmax}, we have: 
\begin{equation}\label{eq:argg}
|x - H_k| \leq  \left(1-\frac{1-d}{n}\right)^k \approx e^{-(1-d)\frac k n}\text{.}
\end{equation}

\end{theorem}

The proof is immediate, as by construction we have $ |F_k|\leq (1-\frac{1-d}{n})|F_{k-1}| $.

Note that the Theorem proves the convergence of \texttt{DI-argmax}, which is not a fair scheduler: for instance, after some time, it will ignore the transient nodes of the graph, which eventually have no fluid left. We conjecture that \eqref{eq:argg} is not tight (tightness would require to always choose an average node) and that the actual convergence should be faster. 



\subsection{Update equation}

The existing iterative methods (such as Gauss-Seidel iteration, Power Iteration, ...) can naturally adapt the iteration when $ G $ (and thus $ P $) is changed because they are in general independent of the initial condition (for any starting vector, the iterative scheme converges to the unique limit). The simplest way to adjust is to compute the PageRank of the new graph using the previous computation as starting vector.

This technique cannot be used in the case of DI so we need to provide an adapted result.

\begin{theorem}
Assume that after $ k_0 $ diffusions, the DI algorithm has computed the values $ (H_{k_0},F_{k_0}) $ for some matrix $ P $, and consider a new matrix $ P'$ (typically an update of $ P $). One can compute the unique solution of the equation $ x'=dP'x'+(1-d)Z $ by running a new D-Iteration with starting parameters
\begin{eqnarray}
F'_0 &=& F_{k_0}+d(P'-P)H_{k_0}\text{,}\\
H'_0 &=& H_{k_0}\text{.}
\end{eqnarray}
\label{thm:update}
\end{theorem}

A few remarks on Theorem \ref{thm:update}:
\begin{itemize}
\item It implies that one can continue the diffusion process when $P$ is regularly updated: we just need to inject in the system a fluid quantity equal to $d(P' - P) H_{k_0}$ and then change to the new matrix $P'$, keeping the same history. 
\item The precision of the result directly relates to the quantity of fluid left. Here the precision induced by $ F_{k_0}+d(P'-P)H_{k_0} $ seems rather minimal, as the original fluid is only altered by the difference between the two matrices. In particular, if the difference $ P-P' $ is small, the update should be quickly absorbed.
\item For sake of clarity, we assumed that the set of nodes is the same between $ P $ and $ P' $, but the result can be extended to cope with variations of $ V $.
\item One can notice that this update may introduce negative fluids (this is indeed mandatory to correct values that are higher in $ H_{k_0} $ than in $x'$), but this has no practical impact on computation.
\end{itemize}

\begin{proof}
Call $ H_{\infty}' $ the asymptotic result of the new D-Iteration. We first use \eqref{eq:HnFn} on the reduced history $ H'_k-H_{k_0} $ (Equation \eqref{eq:HnFn} requires that the history is initially empty):
$$(H'_k-H_{k_0}) + F'_k = F'_0 + d P' (H'_k-H_{k_0}).$$

Letting $ k $ go to $ \infty $ leads to 
$$
\begin{array}{rl}
(H'_{\infty}-H_{k_0}) &= F'_0 + d P' (H'_{\infty}-H_{k_0})\\
&= F_{k_0}+d(P'-P)H_{k_0}+ d P' (H'_{\infty}-H_{k_0})\\
&=  d P' H'_{\infty} +  F_{k_0} -d P H_{k_0},
\end{array}
$$
which can be written

$$H'_{\infty} =d P' H'_{\infty} +  H_{k_0} + F_{k_0} -d P H_{k_0}\text{.} $$

Equation \eqref{eq:HnFn} ($H_{k_0} + F_{k_0} = F_0 + d P H_{k_0}.
$) concludes the proof.

\end{proof}

\section{Numerical evaluation}

%
%
%
%

\label{sec:experiments}
In this section, we compare the convergence speed of DI with other algorithms exposed in Section \ref{sec:the-pagerank-challenge}: Power Iteration, Gauss-Seidel, and OPIC.

The results shown in this paper are based on the 
\texttt{uk-06-07} Web graph, which is compressed using techniques in \cite{BRSLLP},\cite{BV03} and available at \cite{webgraph}. This is a crawl done for DELIS project \cite{DELISProj}, a one-year snapshot (from May 2006 to May 2007) of the \texttt{.uk} domain. It contains 133 million nodes and  5.5 billion links.

\subsection{Settings}

All the algorithms are tuned to solve Equation \eqref{eq:pagerank-pf} using $ d=0.85 $ and $ Z\equiv \frac{1}{n} $.

In the case of OPIC, remind that the original version does not take into account of damping factor. In order be consistent with other algorithms, we emulate \eqref{eq:pagerank-pf} by running OPIC on the stochastic matrix 
$ P' $ defined by:
\begin{equation}
P'_{i,j}=\left\{ 
\begin{array}{l}
\frac{d}{\textrm{out}(j)}+\frac{1-d}{n} \text{ if $ (j,i) \in E$,}\\
\frac{1}{n} \text{if $ j $ is a dangling node,}\\
\frac{1-d}{n} \text{ otherwise.}
\end{array}
\right.
\label{eq:def-popic}
\end{equation}
Note that this emulation makes each diffusion rather costly as all entries need to be updated. It is only introduced to allow comparison with other methods, assuming all diffusions have the same cost, and should not be used in practice.

For DI, we used the two exposed variants, \texttt{DI-cyc} and \texttt{DI-argmax}. The same schedulers were used for OPIC, called 
\texttt{OPIC-cyc} and \texttt{OPIC-argmax}.
%
%
%
Remember that the fluid amount $F$ is constant in OPIC, so the threshold that triggers diffusion in $\texttt{OPIC-argmax}$ is constant (it is the average fluid).
%
\begin{figure}
\begin{center}
  
  
%
%
%
\includegraphics[width=0.80\textwidth]{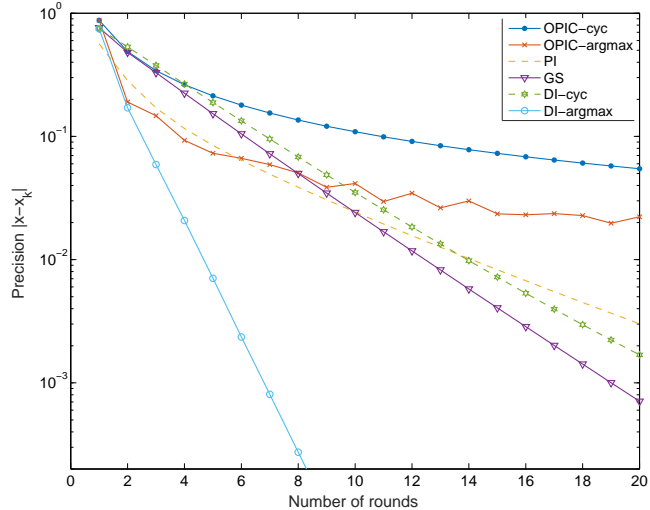}
  
  \caption{Convergence to the stationary PageRank vector}
  \label{fig:exp}
 \end{center}
\end{figure}

\subsection{Results}

The results are shown in Figure \ref{fig:exp}.


The $x$-axis counts the number of rounds required used by the algorithms. For pull methods like PI and GS, a round corresponds to updating the $ n $ entries. For push methods like DI and OPIC, a round corresponds to $ n $ elementary diffusion steps.


The $y$-axis indicates the precision of the current vector $x_l$ at $l^{th}$ round compared to the real PageRank vector $x$ (the distance to the limit or $|x-x_l|$). With DI, the precision can be directly deduced from $|F_n/(1-d-de_n)|$. For other methods, a $10^{-9}$ approximation of $x$ was precomputed using DI.

The OPIC methods perform well during the first few rounds, with a clear advantage of \texttt{OPIC-argmax} over \texttt{OPIC-cyc}. However, the convergence slows down really hard after. Note that OPIC remains interesting as its primary goal was to provide lightweight PageRank estimates. The results only state that OPIC should not be used for precise PageRank estimations

Power Iteration has asymptotically the smallest convergence speed after OPIC, but a good starting performance. As a result, it remains a good candidate for rough estimates (up to $ 10^{-2} $), but other methods should be used if one needs to be more accurate.

The Gauss-Seidel method has the second best convergence speed, although it does not benefit from the quick start observed with PI. It still is a good candidate if one needs a simple and efficient way to compute a precise PageRank.

\texttt{DI-cyc} performs very similarly to Gauss-Seidel (although a little slower). This is in line with our interpretation that \texttt{DI-cyc} is a kind of \emph{push} version of Gauss-Seidel.

Lastly, \texttt{DI-argmax} clearly outperforms the other methods, being in par with \texttt{OPIC-argmax} after two rounds and achieving in 7 rounds the precision reached by Gauss-Seidel in 20 rounds. This shows that \texttt{DI-argmax} is a method of choice for all intended precisions, and suggests that Equation \eqref{eq:argg} really underestimates its actual performance.

\section{Conclusion}
\label{sec:conclusion}
In this paper, we proposed an algorithm based on a diffusion approach, to solve the PageRank equation. We demonstrated some theoretical results concerning the correctness (convergence), the precision measurement and update equations. Our algorithm shows its potential through experiments on real data in comparison with other classical pull (Power Iteration, Gauss-Seidel) and push (OPIC) methods. For future works, we plan to focus on refining the theoretical convergence results and on how to adapt and implement DI in a distributed manner. 

\bibliographystyle{splncs03}
\bibliography{trace}

%
%
%
%
%
%
%

%
\end{document}